\documentclass[english,oneside,aps,preprint]{revtex4}

\usepackage{times,amsmath,amsfonts,amsthm}
\newtheorem{theorem}{Theorem}
\newtheorem{lemma}{Lemma}
\newtheorem{remark}{Remark}

\newcommand{\be}{\begin{equation}}
\newcommand{\ee}{\end{equation}}
\newcommand{\ben}{\begin{eqnarray}}
\newcommand{\een}{\end{eqnarray}}

\usepackage{babel}
\begin{document}

\title{Central limit theorem, deformed exponentials and superstatistics}

\author{C. Vignat$^1$ and A. Plastino$^2$}

\thanks{The authors thank S. Umarov for providing a draft version of \cite{umarov}}

\address{$^1$L.T.H.I., E.P.F.L., Lausanne, Switzerland}

\address{$^2$Facultad de Ciencias Exactas, Universidad Nacional de La Plata and
CONICET, C.C. 727, 1900 La Plata, Argentina }

\email{vignat@univ-mlv.fr, plastino@uolsinectis.com.ar}

\begin{abstract}
We show that there exists a very natural, superstatistics-linked
{\it extension} of the central limit theorem (CLT) to deformed
exponentials (also called q-Gaussians): This generalization
favorably compares with the one provided by S. Umarov and C.
Tsallis  [arXiv:cond-mat/0703533], since the latter requires
a special "q-independence" condition on the data. On the contrary,
our CLT proposal applies exactly in the usual conditions in which
the classical CLT is used. Moreover, we show that, asymptotically,
the q-independence condition is naturally induced by {\it our}
version of the CLT.
\end{abstract}
\maketitle

\section{Introduction}

The central limit theorems (CLT) can be ranked among the most important
theorems in probability theory and statistics and plays an essential
role in several basic and applied disciplines, notably in statistical
mechanics. Pioneers like A. de Moivre, P.S. de Laplace, S.D. Poisson,
and C.F. Gauss have shown that the Gaussian function is the attractor
of independent additive contributions with a finite second variance.
Distinguished authors like Chebyshev, Markov, Liapounov, Feller, Lindeberg and L\'{e}vy have also made essential contributions to the CLT-theory.

The random variables to which the classical CLT refers are required
to be independent. Subsequent efforts along CLT lines have established
corresponding theorems for weakly dependent random variables as well
(see some pertinent
references in \cite{umarov,umarov2,umarov3}). However,
the CLT does not hold if correlations between far-ranging random variables
are not negligible (see \cite{Dehling}).

Recent developments in statistical mechanics that have attracted
the attention of many researches deal with strongly correlated
random variables (\cite{uno} and references therein). These
correlations do not rapidly decrease with any increasing
\textit{distance} between random variables and are often referred
to as global correlations (see \cite{GM} for a definition). 
Is there an attractor that would replace the Gaussians in such a
case?

The answer is in the affirmative, as shown in
\cite{umarov,umarov2,umarov3}, with the deformed or q-Gaussian
playing the starring role. It is asserted in \cite{umarov2} that
such a theorem cannot be obtained if we rely on classic algebra:
it needs a construction based on a special algebra, which is
called q-algebra \footnote{This should not be confused with quantum algebra as defined in \cite{Kac} }. {\it The goal of this communication} is to show that a q-generalization of the central limit theorem becomes
indeed possible
and in a very simple way without recourse to
q-algebra.

\subsection{Systems that are q-distributed}

Consider a system $\mathcal{S}$ described by a random vector $X$
with $d-$components whose covariance matrix reads \be
\label{covar} K=\langle XX^{t}\rangle \equiv EXX^{t}, \ee the
superscript $t$ indicating transposition. We say that $X$ is
$q-$Gaussian (or deformed Gaussian-) distributed if its
probability distribution function writes as described by Eqs.
(\ref{eq:q>1gaussian})-(\ref{eq:Kq>1}) below.

\begin{itemize}
\item in the case $1<q<\frac{d+4}{d+2}$ \begin{equation}
f_{X,q}\left(X\right)=
\frac{\Gamma\left(\frac{1}{q-1}\right)}{\Gamma\left(\frac{1}{q-1}
-\frac{d}{2}\right)\vert\pi\Lambda\vert^{1/2}}
\left(1+X^{t}\Lambda^{-1}X\right)^{\frac{1}{1-q}},\label{eq:q>1gaussian}\end{equation}
with  matrix $\Lambda$ being related to $K$ in the fashion
 \begin{equation}
\Lambda=\left(m-2\right)K.\label{eq:Kq>1}\end{equation}
The number of degrees  of freedom $m$ is defined in terms
 of the dimension $d$ of $X$ as \cite{vignat1}
\begin{equation}
m=\frac{2}{q-1}-d.\label{eq:mq>1}\end{equation}
\item in the case $q<1$ \begin{equation}
f_{X,q}\left(X\right)=\frac{\Gamma\left(\frac{2-q}{q-1}+\frac{d}{2}\right)}
{\Gamma\left(\frac{2-q}{1-q}\right)
\vert\pi\Sigma\vert^{1/2}}
\left(1-X^{t}\Sigma^{-1}X\right)_{+}^{\frac{1}{1-q}},
\label{eq:q<1gaussian}\end{equation} where the matrix $\Sigma$ is
related to the covariance matrix via $\Sigma=pK$. We
introduce here a parameter $p$
 defined as \begin{equation}
p=2\frac{2-q}{1-q}+d,\label{eq:pq<1}.
\end{equation}
\end{itemize}

\section{The road towards a new CLT}

As stated above, several attempts to generalize the central limit
theorem (CLT) have been published recently
\cite{umarov,umarov2,umarov3}, the aim being to have the Gaussian attractor replaced by the
q-Gaussian attractor. We recall here the standard multivariate
version of the CLT.
\begin{theorem}
\label{thm:CLT}  Let $X_{1},X_{2},\dots$ be independent and
identically distributed (i.i.d.) random vectors in
$\mathbb{R}^{d}$ with expectation $E\left[X_{i}\right]=0$ and
covariance matrix $E\left[X_{i}X_{i}^{t}\right]=K$ and let
\begin{equation}
W_{n}=\frac{1}{\sqrt{n}}\sum_{i=1}^{n}X_{i}\label{eq:sum}.\end{equation}
Then $W_n$ converges weakly to a Gaussian vector $W$ with
covariance matrix $K$, or equivalently stated \footnote{Note that
inequality between vectors $W_{n}\le\mathbf{t}$ denotes the set of
$d$ component-wise inequalities $\left\{ W_{n}(k)\le t_{k};1\le
k\le d  \right\}.$} \be \forall
\mathbf{t}\in\mathbb{R}^{d},\lim_{n\rightarrow+\infty}\Pr\left\{
W_{n}\le \mathbf{t}\right\}
=\Phi_{1}\left(\mathbf{t}\right)=\frac{1}{\vert2\pi
K\vert^{1/2}}\int_{-\infty}^{t_{1}}\dots\int_{-\infty}^{t_{d}}e^{-\frac{X^{t}K^{-1}X}{2}}dX.\ee
\end{theorem}
The basic idea leading towards {\it non-conventional} CLTs is to
find conditions under which convergence to the  usual normal
cumulative density function (cdf) $\Phi_{1}$ with covariance
matrix $K$ can be replaced by convergence to a  $q-$Gaussian
cdf
\begin{equation} \Phi_{q}\left(\mathbf{t}\right)=
\int_{-\infty}^{t_{1}}\dots\int_{-\infty}^{t_{d}}f_{X,q}
\left(x\right)dx_{1}\dots dx_{d}\label{eq:Phiq>1}\end{equation}
 with $q>1$, $f_{X,q}$ as defined in (\ref{eq:q>1gaussian}) and
parameter $m$ defined by (\ref{eq:mq>1}) or, for $q<1,$\begin{equation}
\Phi_{q}\left(\mathbf{t}\right)=
\int_{-\infty}^{t_{1}}\dots\int_{-\infty}^{t_{d}}f_{X,q}
\left(x\right)dx_{1}\dots dx_{d}\label{eq:Phiq<1}\end{equation}
 with $f_{X,q}$ as defined in (\ref{eq:q<1gaussian}) and parameter
$p$ defined by (\ref{eq:pq<1}).
We note that both cases $m\rightarrow+\infty$
and $p\rightarrow+\infty$ correspond to convergence $q\rightarrow1$
to the Gaussian case.

In two recent contributions, S. Umarov and C. Tsallis 
highlight the existence of such a central limit theorem,  in the
univariate \cite{umarov2} and multivariate \cite{umarov} case,
provided there exists a certain kind of dependence, called
$q-$independence, between random vectors $X_{i}$. This
$q-$independence condition is expressed in terms of 
the notions of  $q-$Fourier transform $F_{q}$ and of
$q-$product $\otimes_{q}$ \cite{umarov,umarov2} as
\[
F_{q}\left[X_{1}+X_{2}\right]=F_{q}\left[X_{1}\right]\otimes_{q}F_{q}\left[X_{2}\right]\]
which reduces  to conventional independence for $q=1.$

We recall that the $q-$product of $x \in \mathbb{C}$ and $y \in \mathbb{C}$ is\[
x\otimes_{q}y=\left(x^{1-q}+y^{1-q}-1\right)^{\frac{1}{1-q}}\]
 and the $q-$Fourier transform of a function $f\left(x\right),x\in\mathbb{R}^{d},$
is \[
F_{q}\left[f\right]\left(\xi\right)=\int_{\mathbb{R}^{d}}
\left(f^{1-q}\left(x\right)+\left(1-q\right)ix^{t}\xi\right)^{\frac{1}{1-q}}dx.\]
However, this approach suffers from the lack of physical interpretation for
such special dependence; moreover, the $q-$Fourier transform is a
nonlinear transform (unless $q=1$) what makes its use rather difficult.

Another approach, as described in \cite{VP}, consists in keeping
the independence assumption between vectors $X_{i}$ while
replacing the $n$ terms in (\ref{eq:sum}) by a random number
$N\left(n\right)$ of terms. That is,  if the random variable
$N\left(n\right)$ follows a negative binomial distribution so as
to diverge in a specified way, then convergence to a $q-$Gaussian
distribution occurs whenever convergence occurs in the usual
sense.

\textit{In the present  contribution} we show that there exists a
much more natural way to extend the CLT, based on the Beck-Cohen
notion of superstatistics \cite{B1} (see the discussion in
\cite{NC}). Our starting point is the same as that
in Umarov's approach (i.e., assuming some kind of dependence
between the summed terms). However, the manner in which we
introduce this dependence among data is a natural one that can be
interpreted in the physical framework of the Cohen-Beck physics (see \cite{beck} for an interesting overview).

\section{Present results}
Our present results can be conveniently condensed by stating two
theorems, according to the value of parameter $q$. The essential
idea is that of suitably  introducing a chi-distributed random variable
$a$ that is independent (case $q>1$) or dependent (case $q<1$) of the data $X_{i}$,
and then constructing the following scale mixture (typical of
superstatistics \cite{NC}) \be \label{supers}
Z_{n}=\frac{1}{a\sqrt{n}}\sum_{i=1}^{n}X_{i}. \ee

\subsection{The case $q>1$}
\begin{theorem}
\label{thm:CLTq>1}
If
$X_{1},X_{2},\dots$ are i.i.d. random vectors in $\mathbb{R}^{d}$
with zero mean and covariance matrix $K$, and if $\,\,a\,\,$
denotes a random variable chi-distributed with $m$ degrees of
freedom,
  scale parameter $(m-2)^{-1/2}$, and chosen independent of the $X_{i}$,
then random vectors \begin{equation}
Z_{n}=\frac{1}{a\sqrt{n}}\sum_{i=1}^{n}X_{i}\label{eq:Zn}\end{equation}
converge weakly to a multivariate $q-$Gaussian vector $Z$ with
covariance matrix $K.$ Equivalently stated: \be
\forall\mathbf{t}\in\mathbb{R}^{d},
\lim_{n\rightarrow+\infty}\Pr\left\{ Z_{n}\le\mathbf{t}\right\}
=\Phi_{q}\left(\mathbf{t}\right);\ee
 with cdf
  $\Phi_{q}\left(\mathbf{t}\right)$ defined as in (\ref{eq:Phiq>1}).
Moreover, \be q=\frac{m+d+2}{m+d}>1.\ee
\end{theorem}

\begin{proof}
First we note that the $\chi-$density with $m$ degrees of freedom
and scale parameter $\frac{1}{\sqrt{m-2}}$ is\[
f_{a}\left(a\right)=
\frac{2^{1-\frac{m}{2}}\left(m-2\right)^{\frac{m}{2}}}{\Gamma\left(\frac{m}{2}\right)}
a^{m-1}e^{-\frac{a^{2}\left(m-2\right)}{2}}.\] Now, by the
multivariate central limit theorem {\bf \ref{thm:CLT}} above
\footnote{note that, below,  symbol $\Rightarrow$ denotes weak
convergence}
 \[
\frac{1}{\sqrt{n}}\sum_{i=1}^{n}X_{i}\Rightarrow N\]
 where $N$ is a normal vector in $\mathbb{R}^{d}$ with covariance
matrix $K$. Applying from \cite{Bill} its result [Th. 2.8] we deduce that
\[
Z_{n}\Rightarrow\frac{N}{a}\]
where $\frac{N}{a}$  follows a q-Gaussian distribution with covariance matrix $K$
and parameter $q$ defined by (\ref{eq:mq>1}).
\end{proof}

\subsection{The case $q<1$}
The extension of theorem {\bf \ref{thm:CLTq>1}} to the case $q<1$
proceeds as follows.

\begin{theorem}
If $X_{1},X_{2},\dots$are i.i.d. random vectors in
$\mathbb{R}^{d}$ with zero mean and covariance matrix $K$, and if
$\,\,a\,\,$ is  a random variable  independent of the $X_{i}$ that
is  chi-distributed  with $m$ degrees of freedom and scale
parameter $\sqrt{m-2}$, then the  random
vectors
\begin{equation} Y_{n}=
\frac{1}{b \sqrt{n}} \sum_{i=1}^{n}X_{i}
\label{eq:Yn}
\end{equation}
with
\begin{equation}
b=
\sqrt{a^{2}+ \left(\frac{1}{\sqrt{n}}
\sum_{i=1}^{n}X_{i}\right)^{t}\Lambda^{-1}\left(\frac{1}{\sqrt{n}}
\sum_{i=1}^{n}X_{i}\right)}
\end{equation}
converge weakly to a multivariate $q-$Gaussian vector $Y$ with
covariance matrix $K$ and distribution function given by
(\ref{eq:Phiq>1}). Moreover, \be q=\frac{m-4}{m-2}<1.\ee

\end{theorem}
\begin{proof}
If $Z$ has a characteristic distribution function (cdf) given by
(\ref{eq:Phiq>1}), then \cite{VP}\[
Y=\phi\left(Z\right)=\frac{Z}{\sqrt{1+Z^{t}\Lambda^{-1}Z}}\]
 has cdf given by (\ref{eq:Phiq<1}). Since the function
 $\phi=\mathbb{R}^{d}\rightarrow \{Y\in\mathbb{R}^d \, \vert Y^t\Lambda^{-1}Y\le1\}$
is continuous, the desired result is deduced by application of the
continuous mapping theorem (see from \cite{van} its Theorem 2.3,
p.7).
\end{proof}

\begin{remark}
We note that $Y_{n}$ in (\ref{eq:Yn}) is a
normalized version of $Z_n$ in (\ref{eq:Zn}); however, the
fluctuation term $a$ is replaced by a fluctuation term
\[b=\sqrt{a^{2}+\left(\frac{1}{\sqrt{n}}
\sum_{i=1}^{n}X_{i}\right)^{t}\Lambda^{-1}\left(\frac{1}{\sqrt{n}}\sum_{i=1}^{n}X_{i}\right)}\]
that involves the value of the sum itself - and thus is not
independent of this sum anymore. Thus the case $q<1$ can be
considered as a fluctuating version of the usual CLT for which the
fluctuation depends of the state of the system. Moreover, it is
clear that as $n$ increases, the distribution of the fluctuation
$b$ gets closer to a chi distribution with $m+d$ degrees of
freedom.
\end{remark}

\subsection{Link with $q-$independence}

Although the extension of the CLT proposed above differs from the
ones developed  in \cite{umarov}, a link can be established
between both approaches  for large values of $n$ and for
$q>1$ as follows. Note that {\it we assume $q>1$ in the rest of
the paper}.

\begin{theorem}
(linking theorem)
 \label{thm:qindependence}
Assume $1<q<1+\frac{2}{d}$. Consider $n=n_0+n_1$ together with
the division of sum $Z_{n}$ in (\ref{eq:Zn}) into two parts as \be
Z_{n}=\frac{1}{a\sqrt{n}} \left(\sum_{i=1}^{n_{0}}X_{i}
+\sum_{i=n_{0}+1}^{n}X_{i}\right)=Z_{n}^{\left(1\right)}+Z_{n}^{\left(2\right)}.\ee
Assume that the characteristic function $\phi$ of $X_{i}$ is such
that
 $\int_{\mathbb{R}^{d}}\vert\phi\vert^{\nu}dt<\infty$
for some $\nu\ge1,$ and that data $X_i$ are  symmetric ($X_i$ and
$-X_i$ have the same distribution). Then random vectors
$Z_{n}^{\left(1\right)}$ and $Z_{n}^{\left(2\right)}$ are
asymptotically $q-$independent in the sense that\[
\forall\epsilon>0, \, \exists N \,\, {\rm such\,\, that} \,\,
n_{0}>N, n_{1}>N \Rightarrow \vert\vert F_{q} [Z_{n}^{(1)}
+Z_{n}^{(2)}]-F_{q} [Z_{n}^{(1)}]
\otimes_{q_1}F_{q}[Z_{n}^{(2)}]\vert\vert_{\infty}<\epsilon\] with
$q_1=z(q)=\frac{2q+d(1-q)}{2+d(1-q)}.$
\end{theorem}
For didactic reasons we postpone the proof of this result until next Section.  We deduce from it that, asymptotically, the CLT
theorem (\ref{thm:CLTq>1}) exactly generates the $q-$independence
condition required for application of the particular CLT version
proposed in \cite{umarov,umarov2}.

\section{Proof of the linking theorem}

\subsection{Introduction}

In order to simplify the proof we will assume that vectors $X_{i}$
verify a stronger version  of the CLT  than the one stated in theorem
\ref{thm:CLT}, namely the CLT in total variation.
Now, the total variation divergence between two probability
densities $f$ and $g$ is \be
d_{TV}\left(f,g\right)=\frac{1}{2}\int_{\mathbb{R}^{d}}\vert f-g\vert.\ee
 If $U$ and $V$ are random vectors distributed according to $f$
and $g$ respectively, we will denote\[
d_{TV}\left(U,V\right)=d_{TV}\left(f,g\right).\]
 The total variation version of the CLT writes as follows
 (see \cite{van} Th. 2.31.)

\begin{theorem} (CLT in total variation)
\label{thm:dtvclt} Assume that $X_{1},X_{2},\dots$ are i.i.d
random vectors of $\mathbb{R}^{d}$ with zero expectation,
covariance matrix $K$ and characteristic function $\phi$ such that
$\int\vert\phi\vert^{\nu}dt<\infty$ for some $\nu\ge1.\,\,$ If
$W_{n}=\frac{1}{\sqrt{n}}\sum_{i=1}^{n}X_{i}$ and $W$ is a normal
vector in $\mathbb{R}^{d}$ with covariance matrix $K$ then\be
\lim_{n\rightarrow+\infty}d_{TV}\left(W_{n},W\right)=0.\ee
\end{theorem}

Let us introduce the following notations: $\tilde{Z}_{n}$ denotes
a version of sum (\ref{eq:Zn}) where all $X_{i}$ are replaced by
i.i.d. Gaussian vectors $N_{i}\in\mathbb{R}^{d}$ with covariance
matrix $K$, i.e.,  \[ \tilde{Z}_{n}=\frac{1}{a\sqrt{n}}
\left(\sum_{i=1}^{n_{0}}N_{i}+ \sum_{i=n_{0}+1}^{n}N_{i}\right)=
\tilde{Z}_{n}^{\left(1\right)}+\tilde{Z}_{n}^{\left(2\right)}\]
The proof of theorem \ref{thm:qindependence} is based on the fact
that vectors $\tilde{Z}_{n}^{\left(1\right)}$ and
$\tilde{Z}_{n}^{\left(2\right)}$  are exactly $q-$independent (as
seen in subsection \ref{sub:qGaussianareqindependent} below).
Since $n$ is large, according to the above total variations
theorem \ref{thm:dtvclt}, $Z_{n}^{\left(1\right)}$ and
$Z_{n}^{\left(2\right)}$ are close to their $q-$Gaussian
counterparts $\tilde{Z}_{n}^{\left(1\right)}$ and
$\tilde{Z_{n}}^{\left(2\right)}$, respectively (see Lemma IV B
below). It remains to check that closeness between these vectors
can be stated in terms of their $q-$transforms. We proceed in five
steps, that invoke technical lemmas that are the subject of
Subsection C below. These steps are:

\begin{itemize}
\item step 1: components
$\tilde{Z}_{n}^{\left(1\right)}$ and
$\tilde{Z_{n}}^{\left(2\right)}$are exactly $q-$independent, as is
proved in  Thm. \ref{thm:subvectors} of subsection
\ref{sub:qGaussianareqindependent} below.

\item step 2: let us fix $\epsilon>0,$ and write
\begin{eqnarray*}
& \Vert F_{q}[Z_{n}^{(1)}
+Z_{n}^{(2)}]-F_{q}[Z_{n}^{(1)}]
\otimes_{q_1}F_{q}[Z_{n}^{(2)}]\Vert_{\infty} \cr
 & \le \Vert F_{q}[Z_{n}^{(1)}+Z_{n}^{(2)}]
 -F_{q}[\tilde{Z}_{n}^{(1)}+\tilde{Z}_{n}^{(2)}]\Vert_{\infty}\cr
 & + \, \Vert F_{q}[\tilde{Z}_{n}^{(1)}]
 \otimes_{q_1}F_{q}[\tilde{Z}_{n}^{(2)}]
 -F_{q}[Z_{n}^{(1)}]
 \otimes_{q_1}F_{q}[Z_{n}^{(2)}]\Vert_{\infty}
 \end{eqnarray*}

\item step 3: the first term $\Vert F_{q}[Z_{n}^{(1)}+Z_{n}^{(2)}]
 -F_{q}[\tilde{Z}_{n}^{(1)}
 +\tilde{Z}_{n}^{(2)}]\Vert_{\infty}
  = \Vert F_{q}[Z_{n}] -F_{q}[\tilde{Z}_{n}]\Vert_{\infty}$ can be bounded as follows
\[
\Vert F_{q}[Z_{n}] -F_{q}[\tilde{Z}_{n}]\Vert_{\infty}
  \le 2d_{TV}(Z_n,\tilde{Z}_n) \le 2d_{TV}(X_n,\tilde{X}_n)
 \]
where the first inequality follows from Lemma \ref{lem:dtv} and
the second one from Lemma \ref{lem:dtvmixture} below. Thus a value
$N_1$ can be chosen so that $n_0>N_1$ and $n_1>N_1$ ensure that
this term is smaller than $\frac{\epsilon}{2}.$

\item step 4: the second term $\Vert F_{q}[\tilde{Z}_{n}^{(1)}]
 \otimes_{q_1}F_{q}[\tilde{Z}_{n}^{(2)}]
 -F_{q}[Z_{n}^{(1)}]
 \otimes_{q_1}F_{q}[Z_{n}^{(2)}]\Vert_{\infty}$
can be bounded  by applying Lemma \ref{lem:dtvdtvinequality}: for
a large enough value of $n=n_0+n_1$, say $n>N_2$, we have
\[
\Vert F_{q}[\tilde{Z}_{n}^{(1)}]
 \otimes_{q_1}F_{q}[\tilde{Z}_{n}^{(2)}] -F_{q}[Z_{n}^{(1)}]
 \otimes_{q_1}F_{q}[Z_{n}^{(2)}]\Vert_{\infty}
 \le
 2d_{TV}(Z_{n}^{(1)},\tilde{Z}_{n}^{(1)})+2d_{TV}(Z_{n}^{(2)},\tilde{Z}_{n}^{(2)})
\]
Finally,  from the total variation CLT, there exists a value $N_3$
such that $n_0>N_3$ and $n_1>N_3$ implies that each of both total
variation divergences is smaller than $\frac{\epsilon}{4}$.

\item step 5:  The consideration of $N=\max(N_1,N_2,N_3)$ is then seen
 to prove  the linking theorem
\ref{thm:qindependence} 
\end{itemize}
We turn now our attention to those results that we have used in
this proof.

\subsection{\label{sub:qGaussianareqindependent} Components of
$q-$Gaussian vectors are $q-$independent}

We first begin to check that ``sub-vectors" extracted from
$q-$Gaussian vectors are exactly $q-$independent; this results is
obvious from the fact that, by the CLT given in \cite{umarov} (Thm. 4.1), these sub-vectors can be considered as limit cases of
sequences of $q-$independent sequences. However, the mathematical
verification of this property is of an instructive nature and we
proceed to give it. For readability, we will say that  $X \sim
(q,d)$ if $X$ is  a $q-$Gaussian vector of dimension $d$ and
nonextensivity parameter $q.$
\begin{theorem}
\label{thm:subvectors} If $1<q_0<1+\frac{2}{d}$  and vector
$X=[X_1^t,X_2^t]^t\sim (q_0,2d)$ with parameter $q_0>1$ then
vectors $X_1\sim (q,d)$ and $X_2\sim (q,d)$ and they are
$q-$independent:

\be
\label{correlation}
F_{q}\left[X_{1}+X_{2}\right]=F_{q}\left[X_{1}\right]\otimes_{q_1}F_{q}\left[X_{2}\right]
\ee
with $q=z(q_0)=\frac{2q_0+d(1-q_0)}{2+d(1-q_0)}>1$ and $q_1=z(q)>1.$
\end{theorem}

\begin{proof}
Since $X_1 \sim (q,d)$, we know from the  Corollary 2.3 of
\cite{umarov} that $F_{q}\left[X_{1}\right] \sim (q_1,d).$
Moreover, since $X_1$ and $X_2$ are components of the same
$q-$Gaussian vector, from \cite{VP} we deduce that $X_1+X_2 \sim
(q,d)$ so that $F_{q}\left[X_{1}+X_{2}\right]\sim(q_1,d).$
Finally, it is easy to check that since
$F_{q}\left[X_{1}\right]\sim(q_1,d)$ and
$F_{q}\left[X_{2}\right]\sim(q_1,d)$ then
$F_{q}\left[X_{1}\right]\otimes_{q_1}F_{q}\left[X_{2}\right] \sim
(q_1,d)$. The fact that both terms have same covariance matrices is straightforward, what proves the result.
\end{proof}
We note that $q-$correlation (\ref{correlation}) corresponds to $q-$independence  of the third kind as listed in Table 1 of
\cite{umarov}. We pass now to the consideration of the four Lemmas
invoked in the proof of the linking theorem.
\subsection{Technical lemmas}

As we are concerned with scale mixtures of Gaussian vectors, we
need the following lemma.

\begin{lemma}
\label{lem:dtvmixture} If $U$ and $V$ are random vectors  in
$\mathbb{R}^d$ and $a$ is a random variable independent of $U$ and
$V$ then \be d_{TV}\left(\frac{U}{a},\frac{V}{a}\right)\le
d_{TV}\left(U,V\right).\ee
\end{lemma}
\begin{proof}
The distributions of scale mixtures $U/a$ and $V/a$ write, in
terms of the distributions of $U$ and of $V$, in the fashion \be
f_{U/a}\left(x\right)= \int_{\mathbb{R}^{+}}\frac{1}{a^{d}}
f_{a}\left(a\right)f_{U}\left(\frac{x}{a}\right)da, \,\,
g_{V/a}\left(x\right)=\int_{\mathbb{R}^{+}}\frac{1}{a^{d}}f_{a}
\left(a\right)f_{V}\left(\frac{x}{a}\right)da. \ee
 It thus follows that\begin{eqnarray*}
d_{TV}\left(\frac{U}{a},\frac{V}{a}\right) & = &
\frac{1}{2}\int_{\mathbb{R}^{d}}\vert f_{U/a}\left(x\right)-f_{V/a}\left(x\right)\vert dx\\
 & = & \frac{1}{2}\int_{\mathbb{R}^{d}}\vert\int_{\mathbb{R}^{+}}\frac{1}{a^{d}}f_{a}
 \left(a\right)\left(f_{U}\left(\frac{x}{a}\right)-
 f_{V}\left(\frac{x}{a}\right)\right)da\vert dx\\
 & \le & \frac{1}{2}\int_{\mathbb{R}^{d}}\int_{\mathbb{R}^{+}}\frac{1}{a^{d}}
 f_{a}\left(a\right)\vert f_{U}
 \left(\frac{x}{a}\right)-f_{V}\left(\frac{x}{a}\right)\vert dadx\\
 & = & \frac{1}{2}\int_{\mathbb{R}^{+}}\frac{1}{a^{d}}f_{a}
 \left(a\right)da\int_{\mathbb{R}^{d}}\vert f_{U}\left(z\right)-f_{V}
 \left(z\right)\vert a^{d}dz\\
 & = & \frac{1}{2}\int_{\mathbb{R}^{+}}f_{a}\left(a\right)da
 \int_{\mathbb{R}^{d}}\vert f_{U}\left(z\right)-f_{V}\left(z\right)\vert dz\\
 & = & \frac{1}{2}\int_{\mathbb{R}^{d}}
 \vert
 f_{U}-f_{V}\vert=d_{TV}\left(U,V\right).\end{eqnarray*}
\end{proof}

We  also needed above  the following
\begin{lemma}
\label{lem:function} For $q>1$ and $\Re\left(z\right)\ge0,$ the function
\begin{eqnarray*}
\psi_{q,z}:\mathbb{R}^{+} & \rightarrow & \mathbb{C}\\
x & \mapsto &
\left(x^{1-q}+z\right)^{\frac{1}{1-q}}\end{eqnarray*}
 is a Lipschitz function with unit constant:
 \be
 \vert\psi_{q,z}\left(x_{1}\right)-\psi_{q,z}\left(x_{0}\right)\vert
 \le
 \vert x_{1}-x_{0}\vert,
 \ee
\end{lemma}
\begin{proof}
We have \be
\vert\psi_{q,z}\left(x_{1}\right)-\psi_{q,z}\left(x_{0}\right)\vert\le\sup_{x_{0}\le
x\le x_{1}}\vert\psi_{q,z}'\left(x\right)\vert\vert
x_{1}-x_{0}\vert, \ee
 where \be
\psi_{q,z}'\left(x\right)=\frac{1}{\left(1+zx^{q-1}\right)^{\frac{q}{q-1}}},\ee
 with $\frac{q}{q-1}>0,$ so that, since $x>0$ and $\Re\left(z\right)\ge0,$ \be
\vert\psi_{q,z}'\left(x\right)\vert=\frac{1}{\vert1+zx^{q-1}\vert^{\frac{q}{q-1}}}\le
1.\ee
\end{proof}
Two straightforward consequences of such inequality are the
following lemmas, that we have also used above.

\begin{lemma}
\label{lem:dtv} For any random vectors $U$ and $V,$ if $q\ge1,$
the following inequality holds 
\begin{equation} \Vert
F_{q}\left[U\right]- F_{q}\left[V\right]\Vert_{\infty}\le
2d_{TV}\left(U,V\right).\label{eq:inequality}
\end{equation}
\end{lemma}

\begin{proof}
This result is a  straightforward consequence of inequality
(34) of reference \cite{umarov}. However, an elementary proof
writes as follows: denote $f_{U}$ and $f_{V}$ the respective
probability densities of $U$ and $V$. Then,
$\forall\xi\in\mathbb{R}^{d},$\begin{eqnarray*} & \vert
F_{q}\left[U\right]\left(\xi\right)-F_{q}\left[V\right]\left(\xi\right)\vert
\cr & \le  \int_{\mathbb{R}^{d}}
\vert\left(f_{U}^{1-q}\left(x\right)+
\left(1-q\right)ix^{t}\xi\right)^{\frac{1}{1-q}}-
\left(f_{V}^{1-q}\left(x\right)+\left(1-q\right)ix^{t}\xi\right)^{\frac{1}{1-q}}\vert
dx\end{eqnarray*} As $\Re\left(\left(1-q\right)ix^{t}\xi\right)=0$
and $f_{U}\ge0,$ by lemma \ref{lem:function}, the integrand is
bounded by $\vert f_{U}\left(x\right)-f_{V}\left(x\right)\vert;$
since this holds $\forall\xi\in\mathbb{R}^{d},$ the desired result
follows.
\end{proof}
We remark here that inequality (\ref{eq:inequality}) is a simple
generalization of the well-known $q=1$ case, in which $F_{q=1}$
corresponds to the classical Fourier transform. Thus a well-known
result of the Fourier theory is reproduced, namely
\[ \Vert
F_{1}\left[U\right]-F_{1}\left[V\right]\Vert_{\infty}\le
2d_{TV}\left(U,V\right).\]
As another consequence of lemma \ref{lem:function} we have

\begin{lemma}
\label{lem:dtvdtvinequality} For notational simplicity, let us
denote as
$Z_1=Z_{n}^{(1)},Z_2=Z_{n}^{(2)},\tilde{Z}_1=\tilde{Z}_{n}^{(1)}$
and $\tilde{Z}_2=\tilde{Z}_{n}^{(2)}$ those random vectors defined
in part IV.A. Then, for $n$ large enough,
\[ \Vert F_{q}[Z_{1}](\xi)
\otimes_{q_1}F_{q}[Z_{2}]
(\xi)-F_{q}[\tilde{Z}_{1}]
(\xi)\otimes_{q_1}F_{q}[\tilde{Z}_{2}](\xi)\Vert_{\infty}\le
2d_{TV}(Z_{1},\tilde{Z}_{1})+2d_{TV}(Z_{2},\tilde{Z}_{2}).
\]
\end{lemma}
\begin{proof}

For any $\xi\in\mathbb{R}^{d},$
\begin{eqnarray*}
\vert F_{q}[Z_{1}]
(\xi )\otimes_{q_1} F_{q}[Z_{2}]
(\xi )-F_{q} [\tilde{Z}_{1}](\xi )
\otimes_{q_1} F_{q} [\tilde{Z}_{2}](\xi )\vert \\
\le \vert F_{q}[Z_{1}]
(\xi )\otimes_{q_1} F_{q}[Z_{2}]
(\xi )-F_{q} [\tilde{Z}_{1}](\xi )
\otimes_{q_1} F_{q} [Z_{2}](\xi )\vert \\
+\vert F_{q}[\tilde{Z}_{1}]
(\xi )\otimes_{q_1} F_{q}[Z_{2}]
(\xi )-F_{q} [\tilde{Z}_{1}](\xi )
\otimes_{q_1} F_{q} [\tilde{Z}_{2}](\xi )\vert \\
=\vert\psi_{q_1,F_{q}^{1-q_1} [Z_{2} ]
 (\xi )-1} (F_{q} [Z_{1} ]
 (\xi ) )-\psi_{q_1,F_{q}^{1-q_1}
 [Z_{2} ] (\xi )-1}
 (F_{q} [\tilde{Z}_{1} ] (\xi ) )\vert
\\
+\vert\psi_{q_1,F_{q}^{1-q_1} [\tilde{Z}_{1} ]
 (\xi )-1} (F_{q} [Z_{2} ]
 (\xi ) )-\psi_{q_1,F_{q}^{1-q_1}
 [\tilde{Z}_{1} ] (\xi )-1}
 (F_{q} [\tilde{Z}_{2} ] (\xi ) )\vert
\end{eqnarray*}

Since $\tilde{Z}_{2}$ is  $q-$Gaussian, and since
$1<q<1+\frac{2}{d},$ there exists an $\alpha_2\ge0$ (as given in
equation (15) of reference \cite{umarov}) such that $F_{q}^{1-q_1}
[\tilde{Z}_{2} ] (\xi )-1=\alpha_2(q_1-1)\xi^2$ so that, since
$q_1 > 0$, it follows that $F_{q}^{1-q_1}[\tilde{Z}_{2} ] (\xi
)\ge1.$ From the CLT in total variation, we can choose $n$ large
enough so that $d_{TV}(F_{q} [Z_{2} ],F_{q} [\tilde{Z}_{2} ])$ is
arbitrarily small, which in turns implies, by Lemma \ref{lem:dtv},
that $\vert F_{q} [Z_{2} ](\xi) - F_{q} [\tilde{Z}_{2} ](\xi)
\vert$ is arbitrarily small as well. By continuity of the function
$x \mapsto x^{1-q_1}-1$, and since ${F_{q}[{Z}_{2}]}$ is
real-valued by the symmetry of the data, this ensures that
$F_{q}^{1-q_1} [{Z}_{2} ] (\xi )-1\ge0.$ Thus, the first term can
be bounded using lemma \ref{lem:function} in the fashion \[
\vert\psi_{q_1,F_{q}^{1-q_1} [Z_{2} ]
 (\xi )-1} (F_{q} [Z_{1} ]
 (\xi ) )-\psi_{q_1,F_{q}^{1-q_1}
 [Z_{2} ] (\xi )-1}
 (F_{q} [\tilde{Z}_{1} ] (\xi ) )\vert
\le \vert F_{q} [\tilde{Z}_{1} ] (\xi )-
F_{q} [{Z}_{1} ] (\xi )\vert.
\]
Accordingly, since $\tilde{Z}_1$ is $q-$Gaussian,  there exists
$\alpha_1 \ge 0$ such that $F_{q}^{1-q_1}[\tilde{Z}_{1}](\xi
)-1=\alpha_1(q_1-1)\xi^2$, hence $F_{q}^{1-q_1}[\tilde{Z}_{1}](\xi
)\ge 1.$ Recourse again to lemma \ref{lem:function} yields
\[
\vert\psi_{q_1,F_{q}^{1-q_1}[\tilde{Z}_{1}]
 (\xi )-1} (F_{q} [Z_{2} ]
 (\xi ) )-\psi_{q_1,F_{q}^{1-q_1}
 [\tilde{Z}_{1} ] (\xi )-1}
 (F_{q} [\tilde{Z}_{2} ] (\xi ) )\vert
\le \vert F_{q} [\tilde{Z}_{2} ] (\xi )-
F_{q} [{Z}_{2} ] (\xi )\vert.
\]
Applying now lemma \ref{lem:dtv} to each of both terms above
yields
\[
\vert F_{q}[Z_{1}](\xi )\otimes_{q_1} F_{q}[Z_{2}]
(\xi )-F_{q} [\tilde{Z}_{1}](\xi )
\otimes_{q_1} F_{q} [\tilde{Z}_{2}](\xi )\vert
\le 2d_{TV}(Z_{1},\tilde{Z}_{1})+2d_{TV}(Z_{2},\tilde{Z}_{2}).
\]
As this holds for any value of $\xi \in \mathbb{C}$, the result follows.
\end{proof}

\section{Conclusions}

We have here dealt with non-conventional central limit theorems,
whose attractor is a deformed or q-Gaussian. Based on the
Beck-Cohen notion of superstatistics \cite{B1}, with scale
mixtures relating random variables \`a la Eq. (\ref{supers}), it
has been shown that there exists a very natural extension of the
central limit theorem to these deformed exponentials that
quite favorably compares with the one provided by S. Umarov and C.
Tsallis {[}arXiv:cond-mat/0703533]. This is so because the latter
requires a special ``q-independence condition on the data". On the
contrary, our CLT proposal applies exactly in the usual conditions
in which the classical CLT is used. However, links between ours
and the Umarov-Tsallis treatment have also been established, which
makes the here reported CLT a hopefully convenient tool for
understanding the intricacies of the physical processes described
by power-laws probability distributions, as exemplified, for
instance, by the examples reported in  \cite{uno} (and references
therein).


\begin{thebibliography}{9}


\bibitem{umarov} S. Umarov and C. Tsallis,  {[}arXiv:cond-mat/0703533].

\bibitem{umarov2} S. Umarov, S. Steinberg, C. Tsallis,
{[}arXiv:condmat/0603593].

\bibitem{umarov3} S. Umarov, S. Steinberg, C. Tsallis,
{[}arXiv:condmat/0606040] and  {[}arXiv:condmat/0606038](2006).

\bibitem{Dehling} H.G. Dehling, T.Mikosch, M. Sorensen (editors),
{\it Empirical process techniques for dependent data} (Birkhaeser,
Boston-Basel-Berlin, 2002).

\bibitem{uno} Among literally hundreds of references see, for instance,
M.~Gell-Mann and C.~Tsallis, Eds. \textit{Nonextensive Entropy:
Interdisciplinary applications} (Oxford University Press, Oxford,
2004); A. Plastino and A. R. Plastino, Braz. J. of Phys., {\bf 29}
(1999) 50; C. Vignat, A. Plastino, Phys. Lett. A \textbf{365}
(2007) 370.

\bibitem{GM} C. Tsallis, M. Gell-Mann, and Y. Sato, Proc. Natl. Acad.
Sc. USA {\bf 102} (2005) 15377.

\bibitem{vignat1} C. Vignat, A. Plastino, Physics Letters A {\bf 343}
(2005) 411.

\bibitem{VP}C. Vignat and A. Plastino,  Phys. Lett. A {\bf 360} (2007)
415.

\bibitem{B1} C. Beck and E. G. D. Cohen, Physica A {\bf 322} (2003)
267.

\bibitem{NC} C. Vignat, A. Plastino, and A. R. Plastino, Il Nuovo
Cimento B {\bf 120} (2005) 951.


\bibitem{Bill} P. Billingsley, {\it Convergence of Probability Measures},
Second Edition, Wiley Series in Probability and Statistics, 1999.

\bibitem{van}  A.W. van der Vaart, {\it Asymptotic Statistics} (Cambridge
Series in Statistical and Probabilistic Mathematics, 1998).

\bibitem{Kac} V. Kac and P. Cheung, {\it Quantum Calculus}, Springer, 2001

\bibitem{beck} C. Beck, {\it Superstatistics: theoretical concepts and physical applications}, arXiv:cond-mat.stat-mech/07053832v1, 2007
\end{thebibliography}
\end{document}